%% file: main.tex
\documentclass[conference, 10pt, ﬁnal, letterpaper, twocolumn]{IEEEtran}
\usepackage[utf8]{inputenc}
\usepackage{stmaryrd}
\usepackage{amsfonts}
\usepackage{mathrsfs,amsmath}
\usepackage{amssymb}
\usepackage{makeidx}
\usepackage{mhchem}
\usepackage{multirow}
\usepackage{algorithmic} 
\usepackage{bm}
\usepackage{setspace}
\usepackage{amsthm}
\usepackage{empheq}
\usepackage{booktabs}
\usepackage{authblk}
\usepackage[ruled,linesnumbered,vlined]{algorithm2e}  
\usepackage{color}
\usepackage{enumerate}
\usepackage{makecell}
\usepackage[left=0.65in, right=0.65in, top=0.7in, bottom = 1.02in]{geometry}
\usepackage{tabularx,booktabs}
\usepackage{ifpdf}
 \ifpdf
 \else
 \fi
\usepackage{cite}
\usepackage{threeparttable}
\ifCLASSINFOpdf
   \usepackage[pdftex]{graphicx}
\else
\fi
\usepackage[percent]{overpic}

\usepackage{array} 
\usepackage{tikz,pgfplots,filecontents}
\usepackage{tikz-3dplot}
\usepackage{rotating}
\pgfplotsset{width=8.5cm, height=6.33cm, compat=1.9}

\usepackage{physics}
\usepackage{mathdots}
\usepackage{yhmath}
\usepackage{cancel}
\usepackage{color}
\usepackage{siunitx}
\usepackage{array}
\usepackage{multirow}
\usepackage{amssymb}
\usepackage{gensymb}
\usepackage{tabularx}
\usepackage{extarrows}
\usepackage{booktabs}
\usetikzlibrary{spy, fadings, patterns, shadows.blur, shapes, arrows, quotes, angles, calc, shapes, 3d, decorations.markings}

\usepackage{pgfplotstable}
 
\ifCLASSOPTIONcompsoc 
  \usepackage[caption=false,font=normalsize,labelfont=sf,textfont=sf]{subfig}
\else
  \usepackage[caption=false,font=footnotesize]{subfig}
\fi

\usepackage{cases}

\usepackage{amsmath,amssymb,amsfonts}
\usepackage{tikz}
\usetikzlibrary{patterns}

\usepackage{pgfplots}
\pgfplotsset{compat=1.8}
\usepgfplotslibrary{statistics}

\usepackage{tikz-network} 

\SetCommentSty{mycommfont}

\makeatletter
\let\NAT@parse\undefined
\makeatother

\newtheorem{theorem}{Theorem}

\IEEEoverridecommandlockouts   

\allowdisplaybreaks

\pgfplotsset{every axis legend/.style={%
cells={anchor=west},
inner xsep=3pt,inner ysep=2pt,nodes={inner sep=0.8pt,text depth=0.15em},
anchor=north east,%
shape=rectangle,%
fill=white,%
draw=black,
at={(0.98,0.98)},
font=\footnotesize,
}}

\pgfplotsset{every axis/.append style={line width=0.6pt,tick style={line width=0.8pt}}}

\begin{document}

\title{Robust Divergence Angle for Inter-satellite Laser Communications under Target Deviation Uncertainty}

\author{Zhanwei Yu, Yi Zhao, and Di Yuan}
\affil{Department of Information Technology, Uppsala University, Sweden
\authorcr {\em\{zhanwei.yu; yi.zhao; di.yuan\}@it.uu.se}}

\renewcommand*{\Affilfont}{\small}

\maketitle


\begin{abstract}
Performance degradation due to target deviation by, for example, drift or jitter, presents a significant issue to inter-satellite laser communications. In particular, with periodic acquisition for positioning the satellite receiver, deviation may arise in the time period between two consecutive acquisition operations. One solution to mitigate the issue is to use a divergence angle at the transmitter being wider than that if the receiver position is perfectly known. However, as how the deviation would vary over time is generally very hard to predict or model, there is no clear clue for setting the divergence angle. We propose a robust optimization approach to the problem, with the advantage that no distribution of the deviation need to be modelled. Instead, a so-called uncertainty set (often defined in form of a convex set such as a polytope) is used, where each element represents a possible scenario, i.e., a sequence of deviation values over time. Robust optimization seeks the solution that maximizes the performance (e.g., sum rate) that can be guaranteed, no matter which scenario in the uncertainty set materializes.  To solve the robust optimization problem, we deploy a process of alternately solving a decision maker’s problem and an adversarial problem. The former optimizes the divergence angle for a subset of the uncertainty set, whereas the latter is used to explore if the subset needs to be augmented. Simulation results show the approach leads to significantly more robust performance than using the divergence angle as if there is no deviation, or other ad-hoc schemes.
\end{abstract}

\begin{IEEEkeywords}
Robust Optimization, laser communications, inter-satellite communications.
\end{IEEEkeywords}


\section{Introduction}

Inter-satellite laser communication (ISLC) is an effective technology for high-speed data transmission between satellites \cite{Toy2021}. However, the issue of target deviation presents a significant challenge for ISLC. Any deviation due to, for example, drift or jitter, can lead to weaker signal, degraded data rates, or even fail of communication. Periodic acquisition operations are used in ISLC to mitigate target deviation by scanning the target plane with laser beams \cite{Gue2004}. Immediately after an acquisition, when no deviation is present, a highly focused beam using a very small divergence angle can achieve maximum rate at the receiver detector. However, deviations can arise due to drift and jitter if the gap between two consecutive acquisitions is not short enough. The use of a very small divergence angle is not robust in the presence of deviation, as even a small deviation can reduce the rate essentially to zero. By using a wider divergence angle, the light energy is more spread out, and a positive rate can be achieved if the deviation is within a certain range. In fact, one can calculate the optimal divergence angle even though we only know the deviation of the receiver, instead of the exact position \cite{song2016performance}. However, it is difficult to model or predict the deviation.

This paper proposes a robust optimization approach to address the issue of target deviation in ISLC without assuming any specific distributions. Robust optimization finds solutions being resilient to uncertainty or variation in system parameters \cite{ben2009robust}. The goal is to promise a performance that can be guaranteed across all possible scenarios, rather than for any deterministic. To be specific, in robust optimization, a so called uncertainty set where each element represent a possible scenario  under the uncertain system parameters, and the solution is the optimal in performance while safeguarding against all elements of the set. 

In our problem, we use a time-slotted horizon and an element of the uncertainty set is a sequence of deviation values over time slots. By incorporating the uncertainty set, we formulate a robust optimization problem that aims to maximize the sum rate between two acquisition operations while accounting for the uncertainty represented by the uncertainty set. The objective is to identify a divergence angle that guarantees a performance that is robust against all scenarios in the uncertainty set. This means that regardless of which scenario materializes, the resulting performance is at least as good as the one reported by the robust optimization. Furthermore, the optimal solution of the robust optimization ensures that no other solution with a better performance guarantee exists.

The outlined robust problem is hard to tackle due to the lack of explicit mathematical equations in the optimization problem. We deploy a process of alternately solving a decision maker’s problem (DMP) and an adversarial problem (AP). The process exhibits resemblance to a game-theoretic approach known as the ``minmax algorithm". The DMP corresponds to the ``maximization" step, where we aim to optimize the divergence angle for a subset of the uncertainty set. The AP corresponds to the ``minimization" step, where the task is to examine the worse-case performance of the current DMP solution, and thereby to determine if the subset needs to be augmented. Our simulation results demonstrate that the approach provides a good and robust sum rate in the ISLC system compared to using the divergence angle as if there is no deviation, or other ad-hoc schemes. 


\section{Related Works}

There are some studies on beam divergence angle control to reduce the impact of the target deviation. The authors of \cite{lee2022dynamic} use variable focus lenses at both the transmitter and receiver to mitigate the issues caused by pointing errors and angle-of-arrival fluctuations. The authors of \cite{mai2019beam} propose a rapid and power-efficient adaptive beam control technique using non-mechanical variable-focus lenses to address target deviation due to angle-of-arrival fluctuation and pointing error for outage probability. Via approximating the detector shape of the receiver, the authors of \cite{song2016performance} provide a closed-form of the optimal divergence angle for any given deviation of the receiver. The optimal beam divergence angle that minimizes the average bit error probability is also studied in \cite{toyoshima2002optimum, do2020numerical}.

Our work differs from the works in \cite{lee2022dynamic, mai2019beam} because the proposed systems are assumed to be monitoring the target deviation and then they propose a divergence angle optimization method. In addition, our study is different from the studies in \cite{song2016performance, toyoshima2002optimum, do2020numerical} that require knowledge of the distribution of target deviation to analyze the system's performance. In contrast, our paper does not rely on such assumptions.


\section{System Model and Formulation}

\input{fig-laser}

\subsection{The Gaussian Beam Model} \label{gaussianbeam}

Consider laser communication between two satellites. Fig. \ref{fig:system_model-a} shows that the laser is transmitted with a small divergence angle to the receiver when there is no target deviation. The blue circle in the figure represents the transmitter lens, through which a laser beam is transmitted to the receiver detection plane. The distance between the transmitter and the plane is denoted by $L$, and the transmitter beam divergence angle is represented by $\theta$. The laser beam is modeled as a Gaussian beam, and the intensity of a point $(x,y)$ in the plane, is given by \cite{saleh2019fundamentals}
\begin{equation}\label{integration}
I\left(L, \theta, x, y\right) \overset{\Delta}{=} \frac{2}{\pi L^2\theta^2} \cdot \exp\left(\frac{-2\left(x^2 + y^2\right)}{L^2\theta^2}\right).
\end{equation}
The receiver has a disk-shaped detection area with radius $r$. The received signal strength is proportional to the total received energy of the beam over the detection area. 

In Fig. \ref{fig:system_model-a}, it can be observed that without any deviation, the transmitter can use a small divergence angle to focus the laser over the detection area, resulting in a strong total received energy. However, due to drift and jitter, the center of the detection area might not be at the origin. Fig. \ref{fig:system_model-b} and \ref{fig:system_model-c} illustrate the same target deviation, but the effect of divergence angle, respectively. One can see that, with a relatively small divergence angle $\theta_1$, the total received energy over the detection area is significantly smaller than with $\theta_2$. Note that, however, if the angle is higher than $\theta_2$, the laser footprint will spread out further, and the total received energy might reduce instead. Thus, the beam divergence angle $\theta$ has a significant impact on the total received energy over the receiver's detection area (and hence the rate).
 
\subsection{Time Horizon and Uncertainty}
We consider a horizon of $T$ time slots between two acquisitions. Let $\mathcal{T} \overset{\Delta}{=} \{1, 2, ..., T\}$, and we use $d_t$ $(t \in \mathcal{T})$ to represent the distance between the center of the detector disk and the origin in time slot $t$. We assume that there is no target deviation immediately after acquisition. Over time, deviation $d_t$ $(t \ge 1)$ may occur, resulting in a change in the detector's center. We denote the detector disk area in time slot $t$ as $D_t$. Note that even if $d_t$ is known, the position of the detector remains unknown. However, the function $I(x,y)$ is centrosymmetric. Thus, without loss of generality, we assume that the deviations only occur along the $x$-axis. Thus, $D_t$ can be expressed as:
\begin{align}
    D_t: \left\{\left(x, y\right)| \left(x-d_t\right)^2+y^2 \le r^2\right\}. \label{constraint2}
\end{align}

We represent the deviations in time slots as $\boldsymbol{d} \overset{\Delta}{=} \{d_1, d_2, ..., d_T\}$ $(d_t \ge 0, \forall t \in \mathcal{T})$. Robust optimization uses an uncertainty set to account for uncertain parameters or conditions. A common choice is the so-called ``box" or ``hyper-rectangle" uncertainty set, where each uncertain parameter is allowed to vary within a fixed range or interval. However, box uncertainty is found to be too pessimistic \cite{gorissen2015practical}. A better option is to use a ``budget" constraint that limits the total uncertainty or deviation allowed in the system. This is often motivated by that all uncertain parameters reaching their worst-case bounds simultaneously is extremely unlikely to occur in practice.

We use an uncertainty set $\mathcal{U}$ defined as follows, though we remark that our algorithmic approach is not limited to this particular definition. 
\begin{numcases}{\mathcal{U}:}
    &$ d_1  \le d_{\text{gap}} ,$ \label{constraint4}\\
    &$ \left| d_{t+1} - d_{t} \right| \le d_{\text{gap}}, \forall t \in \mathcal{T} \setminus \left\{T\right\},$ \label{constraint5}\\
    &$ \sum_{t \in \mathcal{T}} d_t \le d_{\text{total}},$ \label{constraint6}
\end{numcases}
\begin{itemize}
    \item Constraint \eqref{constraint4} state that the deviation can be at most $d_{\text{gap}}$ in the first time slot.
    \item Constraint \eqref{constraint5} imposes that the difference between the deviations in two consecutive time slots is at most $d_{\text{gap}}$. This is motivated by that deviation does not change too rapidly over neighboring time slots. In addition, this constraint also sets the worst-case deviation of any individual time slot.
    \item Constraint \eqref{constraint6} is the so-called budget constraint stating the total deviation over all time slots is at most $d_{\text{total}}$ (with $d_{\text{total}} < T \cdot d_{\text{gap}}$).
\end{itemize}

\subsection{Problem Formulation}
  In time slot $t$, the total received energy by the detector can be calculated by
\begin{align}\label{initial_integral}
    \iint_{D_t} I\left(L, \theta, x, y\right)dx dy.
\end{align}
Given transmitter power $P$, transmission optical efficiency $\tau_{opt}$, the achievable data rate of the receiver is computed by
\begin{align}
    R_t = \frac{P\tau_{opt}}{E_p N_b}\iint_{D_t} I\left(L, \theta, x, y\right)dx dy,
\end{align}
where $E_p = \frac{hc}{\lambda}$ is the photon energy and $N_b$ is the receiver sensitivity in photons/bit \cite{majumdar2005free}. 

We would like to optimize the divergence angle to maximize the sum rate $R_{\text{sum}} \overset{\Delta}{=} \sum_{t\in \mathcal{T}} R_t$ that can be guaranteed between two acquisitions. The robust optimization problem can be formulated as
\begin{align}\label{formulation}
    R^* = \max_{\theta}\ \min_{\boldsymbol{d} \in \mathcal{U}}\ & R_{\text{sum}}.
\end{align}

An intuition for looking for a robust solution is to identify a single ``worse-case" scenario of the uncertainty set, and then design the system with respect to that scenario. However, this intuition fails for two reasons. First, even if there is such a scenario (i.e., a vector $\boldsymbol{d}$), identifying it may be very difficult. Second, there may not even exists a single scenario that is the ``worst case". For instance, one sub-range of $\theta$ may have worst-case scenario $\boldsymbol{d}$, while another sub-range may have worst-case scenario $\boldsymbol{d}'$. Thus, it is not viable to approach the problem with a ``worst-case" mind. Additionally, the problem involves integrals, which makes it intractable to model with explicit mathematical equations. Therefore, it is not feasible to apply the reformulation technique \cite{averbakh2008explicit} of robust optimization.

\section{Problem Solving}

\subsection{Overview}

We propose an iterative solution procedure. The approach begins by considering a finite subset of scenarios of the uncertainty set. Next, the corresponding robust optimization problem for this subset is solved. If the solution obtained is robustly feasible, then we have found the optimal solution. The term robustly feasible means that the performance of this solution, even though obtained for the subset, can in fact be guaranteed for the entire uncertainty set. If not, we need to identify a scenario from the uncertainty set that is not robustly feasible. Thus, we look for the scenario in the uncertainty set that results in lowest performance of the current solution. Once this scenario is identified, it is added to the subset, and the expanded optimization problem is solved. This process continues until a robustly feasible solution is found or until a prescribed convergence tolerance is reached.

Specifically, the proposed approach comprises two modules. That are alternately tackled. The first module involves solving a decision maker’s problem (DMP) as follows, which is similar to the original problem formulation but with a finite subset of the uncertainty set. 
\begin{align}\label{Main_problem}
    \text{[DMP] }\max_{\theta}\ \min_{\boldsymbol{d}\in \mathcal{U}'}\ & R_{\text{sum}}(\theta)  
\end{align}
where $\mathcal{U}' \subset \mathcal{U}$ is finite.

Once we obtain the optimal solution $\theta^*_{\text{DMP}}$ of the DMP \eqref{Main_problem} and its objective function value $R_{\text{sum}}(\theta^*_{\text{DMP}})$, we move on to finding the most unfavorable deviation scenario for $\theta^*_{\text{DMP}}$. This task is called the adversarial problem (AP). The AP is formulated as:
\begin{align}\label{Secondary_problem}
    \text{[AP] }\min_{\boldsymbol{d}\in \mathcal{U}}\ & R_{\text{sum}}(\theta^*_{\text{DMP}}, \boldsymbol{d})   \\
    \text{s.t. } & \eqref{constraint4}\text{-}\eqref{constraint6}. \notag
\end{align}
In the AP, we find $\boldsymbol{d}^*_{\text{AP}}$ in $\mathcal{U}$ that is the worst-case deviation sequence, i.e., leading to lowest sum rate, for angle $\theta^*_{\text{DMP}}$. If $R_{\text{sum}}(\theta^*_{\text{DMP}}, \boldsymbol{d}^*_{\text{AP}}) < R_{\text{sum}}(\theta^*_{\text{DMP}})$, angle $\theta^*_{\text{DMP}}$ is not robustly feasible. Thus, we add $\boldsymbol{d}^*_{\text{AP}}$ as a new scenario to $\mathcal{U}'$, and repeat the process. 


\subsection{AP as a Resource Constrained Shortest Path Problem}
We would like to highlight that finding the solution $\boldsymbol{d}^*_{\text{AP}}$ to AP is not an easy task, even if $\theta^*_{\text{DMP}}$ is given. The sum rate expression uses integrals, making the problem computationally intractable. To overcome this challenge, we propose the use of discretization and then reduce the problem to a resource-constrained shortest path problem (RCSPP).

We first discretize the continued variables $d_t$'s ($t\in\mathcal{T}$) of AP by step size $\Delta$, i.e.,
\begin{align}\label{discretize}
    d_t \in \left\{n \cdot \Delta| n \in \mathbb{N} \right\}, \forall t \in \mathcal{T},
\end{align}
where $\mathbb{N}$ contains non-negative integers.
This leads to a discrete approximate AP (AAP). The AAP has the same constraints \eqref{constraint4}-\eqref{constraint6} but its optimization variables follow \eqref{discretize}. We will show that AAP maps to finding the shortest path with resource constraint in a directed graph. The graph is illustrated in Fig. \ref{fig:shortest-path}, and the construction is detailed below.

\input{fig-shortest-path}

\subsubsection{Structure Overview}

The graph consists of $T$ sections indexed by time slot $t$ $(t \in \mathcal{T})$ except the source node $\Phi$ and the sink node $\Omega$. Node $\Phi$ represents the state immediately before time slot one. For $\lambda$-nodes in time slot $T$, they all have an arc to the sink node $\Omega$. Each $\lambda$-node represents a deviation state, where the superscript indicates the associated time slot and the subscript represents the deviation. For instance, node $\lambda^{t}_{i}$ represents that the deviation $d_{t} = i \cdot \Delta$ in time slot $t$. In time slot $t$, there are $1 + t \times \lfloor \frac{d_{\text{gap}}}{\Delta} \rfloor$ $\lambda$-nodes. Arc $(\Phi \rightarrow \lambda ^{1}_{i})$ denotes the deviation changes from $0$ to $i \cdot \Delta$ in time slot one. Similarly, Arc $(\lambda ^{t}_{i} \rightarrow \lambda ^{t+1}_{\hat{i}})$ denotes that the deviation changes from $i \cdot \Delta$ to $\hat{i} \cdot \Delta$ of time slots $t$ and $t+1$. Because the deviation of two consecutive time slots is constrained by \eqref{constraint4}, any node $\lambda ^{t}_{i}$ $(1\le t \le T-1)$ or source node $\Phi$ has at most $1+2 \times \lfloor \frac{d_{\text{gap}}}{\Delta} \rfloor$ arcs to some $\lambda$-nodes in time slot $t+1$ (with $t = 0$ for node $\Phi$ as a special case). Let $\hat{i}$ represent a subscript of these in time slot $t+1$,then $\hat{i}$ is constrained by
\begin{align}\label{gap constraint}
    \hat{i} \in \left[i-\lfloor \frac{d_{\text{gap}}}{\Delta} \rfloor, i+\lfloor \frac{d_{\text{gap}}}{\Delta} \rfloor\right] \cap \mathbb{N}.
\end{align} 

\subsubsection{Arcs Attributes}

In RCSPP, an arc has two attributes, i.e., weight and resource consumption. For arc $(\lambda ^{t}_{i} \rightarrow \lambda ^{t+1}_{\hat{i}})$, the weight is calculated by
\begin{align}\label{wight}
    \frac{P\tau_{opt} }{E_p N_b}\iint_{D_{\hat{i}}} I\left(L, \theta^*_{\text{DMP}}, x, y\right)dx dy,
\end{align}
where
\begin{align}
    D_{\hat{i}}: \left\{\left(x,y\right)|(x-\hat{i} \cdot \Delta)^2+y^2 \le r^2\right\}.
\end{align}
In addition, the resource consumption of the arc is set to be $\hat{i}$. For arc $(\lambda ^{T}_{i} \rightarrow \Omega)$, the weight and the resource consumption are both zero.

\subsubsection{Resource Constraint}

In RCSPP, selecting an arc implies to consume the associated amount of resource, and we need to find the shortest path within a given resource limit. Considering constraint \eqref{constraint6}, the resource limit in the corresponding RCSPP is $\lfloor \frac{d_{\text{total}}}{\Delta} \rfloor$.

\begin{theorem}
For any given $\theta^*_{\text{DMP}}$, the AAP \eqref{Secondary_problem} can be solved in pseudo-polynomial time as an RCSPP. 
\end{theorem}
\begin{proof}
Suppose the optimal solution $\boldsymbol{d}^*_{\text{AP}} = \left\{d_1^*, d_2^* , ..., d_T^*\right\} = \left\{n_1^* \cdot \Delta, n_2^* \cdot \Delta, ..., n_T^*\cdot \Delta\right\}$ is given. The corresponding path is $ \Phi \rightarrow  \lambda^1_{n_1^*} \rightarrow \cdots \rightarrow \lambda^T_{n_T^*} \rightarrow \Omega$. According to our setting, the resource consumption of the path must be less than the resource limit. In addition, the total arc weight of the derived path equals the sum rate for $\boldsymbol{d}^*_{\text{AP}}$ because we set weights by \eqref{wight} that represents the rate of a time slot for a given deviation.

Conversely, consider the shortest path complying with the resource limit. For time slot $t$, if node $\lambda^t_{i}$ is contained in the path, we set $d_t^*=i\cdot\Delta$. With this variable value setting, the objective function has the same value as the path length (i.e., total arc weight). According to \eqref{gap constraint}, $d_t^*$ and $d_{t+1}^*$ must meet constraint \eqref{constraint5}. Similarly, $\sum_{t\in \mathcal{T}}d_t^*$ must meet constraint \eqref{constraint6} due to the resource limit $\lfloor \frac{d_{\text{total}}}{\Delta} \rfloor$ in RCSPP. One can observe that the solution constructed in this way satisfies the constraints of the AAP. Since RCSPP is weakly NP-hard \cite{JOKSCH1966191}, the conclusion follows.
\end{proof}

\subsection{DMP as a Series of Linear Programming Problems}

After solving AAP, we add the corresponding unfavorable scenario to $\mathcal{U}'$, and aim to solve the DMP under this new subset. However, solving DMP can be challenging as it again involves integrals. We propose a linear approximation method to effectively tackle the DMP.

To start, we discretize variable $\theta$ into $M$ equally sized intervals. Specifically, for $\theta \in \left[\alpha, \omega\right]$, we discretize it into $M$ intervals. We use $\left[\alpha_m, \omega_m \right]$ to represent the $m$-th interval $(m \in \mathcal{M} \overset{\Delta}{=} \left\{1,2,...,M\right\})$, where
\begin{align}\label{interval}
    \alpha_m &= \alpha + (m-1)\cdot \frac{\omega - \alpha}{M},\\
    \omega_m &= \alpha + m\cdot \frac{\omega - \alpha}{M}.
\end{align}
For the $m$-th interval, we use a linear function denoted by ${R}_{\text{sum}}^{\boldsymbol{d}} (\theta, m)$ to approximate the sum rate $R_{\text{sum}} (\theta)$ as a function of $\theta$ under a given scenario $\boldsymbol{d} \in \mathcal{U}'$. The function ${R}_{\text{sum}}^{\boldsymbol{d}} (\theta, m)$ is expressed by
\begin{align}
    {R}_{\text{sum}}^{\boldsymbol{d}} (\theta, m) = \frac{R_{\text{sum}} (\omega_m) - R_{\text{sum}} (\alpha_m)}{\omega_m - \alpha_m} \theta + R_{\text{sum}} (\alpha_m).
\end{align}
Note that the sum rates $R_{\text{sum}} (\alpha_m)$ and $R_{\text{sum}} (\omega_m)$ are easily calculated as the divergence angle $\theta$ (i.e., $\alpha_m$ and $\omega_m$) and the deviation vector $\boldsymbol{d}$ are all given. Thus, we have the following approximated DMP in the $m$-th interval:
\begin{align}\label{AMP}
    \max_{\theta \in \left[\alpha_m, \omega_m \right]}\ \min_{\boldsymbol{d}\in \mathcal{U}'}\ & {R}^{ \boldsymbol{d}}_{\text{sum}} (\theta, m).
\end{align}
With an auxiliary variable $R_m$, we transform problem \eqref{AMP} into the following equivalent linear programming problem:
\begin{subequations}\label{AMP2}
\begin{align}
    \max_{\theta, R_m}\ & R_m \\
    \text{s.t. } & R_m \le {R}^{\boldsymbol{d}}_{\text{sum}} (\theta, m), \forall \boldsymbol{d}\in \mathcal{U}'.
\end{align}
\end{subequations}
We solve a series of problems \eqref{AMP2} for all intervals. Let $\theta^*_m$ and $R_m^*$ denote the optimal solution and the optimal sum rate, respectively, for the $m$-th interval. The solution to DMP for $\theta \in \left[\alpha, \omega\right]$ is given by
\begin{align}
    \theta^*_{\text{DMP}} = \underset{\theta^*_m}{\mathrm{argmax}}\, \left\{ R_m^*| m \in \mathcal{M} \right\}.
\end{align}

\subsection{Summary and Convergence}

Note that solving DMP yields an upper bound (UB) that monotonically improves over the iterations, while the AAP yields a lower bound (LB). We solve the robust optimization problem once the difference between UB and best LB reaches a specified tolerance $\epsilon$. The complete approach is summarized in Algorithm \ref{al:framwork}.

\begin{algorithm}[tbp]
    \DontPrintSemicolon
    \caption{Robust optimization for divergence angle} \label{al:framwork}
    \KwIn{$L$, $r$, $T$, $d_{\text{gap}}$, $d_{\text{total}}$, $\mathcal{U}$, $\epsilon$} 
    \KwOut{$\theta^*$}
    Initialize $ \mathcal{U}^\prime$, Flag $\leftarrow 0$, LB $\leftarrow 0$\\
    \Repeat
    {
    	\rm Flag $= 1$\
    }
    {
    	Solve DMP to obtain $\theta^*_{\text{DMP}}$ and $R_{\text{sum}}(\theta^*_{\text{DMP}})$\\
            UB $\leftarrow R_{\text{sum}}(\theta^*_{\text{DMP}})$ \\ 
            Solve AAP to obtain $\boldsymbol{d}^*_{\text{AP}}$ and $R_{\text{sum}}(\theta^*_{\text{DMP}}, \boldsymbol{d}^*_{\text{AP}})$\\
            \If{\rm LB $< R_{\text{sum}}(\theta^*_{\text{DMP}}, \boldsymbol{d}^*_{\text{AP}})$}
            {
                LB $\leftarrow R_{\text{sum}}(\theta^*_{\text{DMP}}, \boldsymbol{d}^*_{\text{AP}})$\\ 
            }
            \uIf{\rm UB $-$ LB $\le \epsilon$}
            {
                Flag $\leftarrow 1$\\ 
            }             
            \Else
            {
                Add $\boldsymbol{d}^*_{\text{AP}}$ into $\mathcal{U}'$\\
            }
    }   
    \Return {$\theta^* \leftarrow \theta^*_{\text{DMP}}$}\\
\end{algorithm} 

\begin{theorem}
    DMP and AAP provide UB and LB to the optimum $R^*$ of \eqref{formulation}, respectively, with UB improving monotonically over the iterations. Additionally, Algorithm \ref{al:framwork} is guaranteed to converge.
\end{theorem}
\begin{proof}
    For any $\theta >0$, we have $\min_{\boldsymbol{d}\in \mathcal{U}'} R_{\text{sum}}(\theta) \ge \min_{\boldsymbol{d}\in \mathcal{U}}R_{\text{sum}}(\theta)$ since $\mathcal{U}'$ is the subset of $\mathcal{U}$. Therefore, for DMP, $R_{\text{sum}}(\theta^*_{\text{DMP}}) \ge R^*$. Next, for any $\theta^*$ obtained from DMP, we have $\min_{\boldsymbol{d}\in \mathcal{U}}\ R_{\text{sum}}(\theta^*_{\text{DMP}}, \boldsymbol{d}) \le R^*$ by the definition of $R^*$. Thus, $R_{\text{sum}}(\theta^*_{\text{DMP}},\boldsymbol{d}^*_{\text{AP}}) \le R^*$ and the AAP provides LB to $R^*$. 

    Because $\mathcal{U}'$ grows over iterations, the UB obtained by DMP improves monotonically. Furthermore, by the construction of AAP, only a finite number of scenarios will be added to $\mathcal{U}'$ as there is a finite number of paths in the graph. Thus, Algorithm \ref{al:framwork} will converge. Hence the conclusion.
\end{proof}

\section{Performance Evaluation}


Table \ref{tab:parameters} provides details of the parameters used in the simulation. Note that deviation in unit of centimeters meters can be calculated by multiplying the angular deviations with the distance $L$. We set the maximum total deviation change $d_{\text{total}}$ to be proportional to the number of time slots $T$. We use the no-deviation scenario (i.e., $\boldsymbol{d} = \boldsymbol{0}$) to initialize the subset $\mathcal{U}'$.

\begin{table}[h]
    \caption{\label{tab:parameters}Simulation Parameters.}
    \footnotesize
    \begin{center}
        \begin{threeparttable}[b]
            \begin{tabular}{*{2}{lr}}
                \toprule
                \midrule
                {\bf Parameter} & {\bf Value}\\
                \midrule
                Divergence angle ($\theta$) & $[0.01, 1000]$ $\mu$rad\\
                Distance ($L$) & $40$ km\\
                Radius of the detector ($r$) & $15$ cm\\
                Transmit power ($P$) & $70$ mW\\   
                Transmission optical efficiency ($\tau_{opt}$) & $0.01$\\
                Laser wavelength ($\lambda$) & $850$ nm\\
                Receiver sensitivity ($N_b$) & $100$ photons/bit\\
                Number of time slots ($T$) & $\{6, 8, 10, 12, 14\}$ \\
                Maximum deviation change in a time slot ($d_{\text{gap}}$) & 1$\mu$rad \\
                Maximum total deviation ($d_{\text{total}}$) & $40\% \times T\times d_{\text{gap}}$ $\mu$rad \\
                Convergence tolerance ($\epsilon$) & $10^{-4}$\\
                \bottomrule
            \end{tabular}
        \end{threeparttable}
    \end{center}
\end{table}

In our simulation, the robust angle (RA) obtained by our approach is compared to the following two reference schemes: 
\begin{itemize}
    \item Small angle (SA): This scheme transmits the laser using the smallest possible angle; SA is optimal for the no-deviation scenario.
    \item Average-deviation Angle (AA): This scheme splits the maximum possible total deviation evenly over $T$ time slots (i.e., $\frac{d_{\text{total}}}{T}$) and then finds the optimal divergence angle for this averaged deviation by the method in \cite{song2016performance}.
\end{itemize}

For each scheme, we compare its worst-case performance, i.e., the scenario in $\mathcal{U}$ leading the lowest sum rate. Note that for RA, the worst-case performance is at most $\epsilon$ lower by algorithm construction. For SA and AA, the worst-case performance is derived by solving the corresponding instances of AAP. 

We also examine the performance in an average sense. To this end, we generate randomly scenarios from set $\mathcal{U}$. For $t \in \mathcal{T}$, $d_t$ is randomly and uniformly generated in $\left[0, d_{\text{total}}\right]$. The scenario is accepted if it satisfies the constraints of the uncertainty set $\mathcal{U}$ (i.e., it is indeed in uncertainty set $\mathcal{U}$). A total of 1000 such random scenarios in $\mathcal{U}$ are used for simulation.

\input{result1}

Fig. \ref{fig:result1} shows the rate per time slot of the three schemes under their respective worst-case scenarios. In other words, the results represent the performance that can be guaranteed by the schemes no matter which scenario in the uncertainty set materializes. The SA scheme always has a rate being close to zero, as it has no safeguard against any deviation. The AA scheme is significantly better in robustness, whereas we obtain about 10\% higher performance guarantee via RA. The rate of both RA and AA decreases gradually in the number of time slots as the (total) deviation accumulates.

\input{result2}

In Fig. \ref{fig:result2}, we present box plots of the simulation results of 1000 random scenarios in $\mathcal{U}$, where outliers are classified using the 1.5 interquartile range rule. SA remains inferior, both in terms of average and the number of outliers of low rate. Note that for some ``lucky” scenarios (i.e., if the deviation happens to be very small), SA does provide very good rate. For RA and AA, one can observe that the latter has a slightly better average rate, with the price of being less overall robustness (as shown in the previous figure), and outliers associated with the unfavorable scenarios.

\input{result3}

In Fig. \ref{fig:result3}, we illustrate algorithm convergence. The plot demonstrates that the proposed method has a satisfactory convergence speed; the UBs and LBs meet after a small number of iterations.

\section{Conclusion}

We proposed a robust optimization approach for maximizing the sum rate for in inter-satellite laser communication under uncertainty of the target deviation. The approach enables a solution that provides the best possible performance guarantee for all scenarios in the uncertainty set. From the simulation results, the average performance of the robust solution is very good (and stable). Thus robust optimization does open new perspectives for satellite laser communications.

\bibliographystyle{IEEEtran}
\bibliography{mybibtex}

\end{document}

%% file: fig-laser.tex
\tikzoption{canvas is xy plane at z}[]{%
  \def\tikz@plane@origin{\pgfpointxyz{0}{0}{#1}}%
  \def\tikz@plane@x{\pgfpointxyz{1}{0}{#1}}%
  \def\tikz@plane@y{\pgfpointxyz{0}{1}{#1}}%
  \tikz@canvas@is@plane
}

\definecolor{ForestGreen}{RGB}{34,139,34}

\pgfmathsetmacro{\X}{0}
\pgfmathsetmacro{\Y}{1.7}
\tdplotsetmaincoords{110}{70}

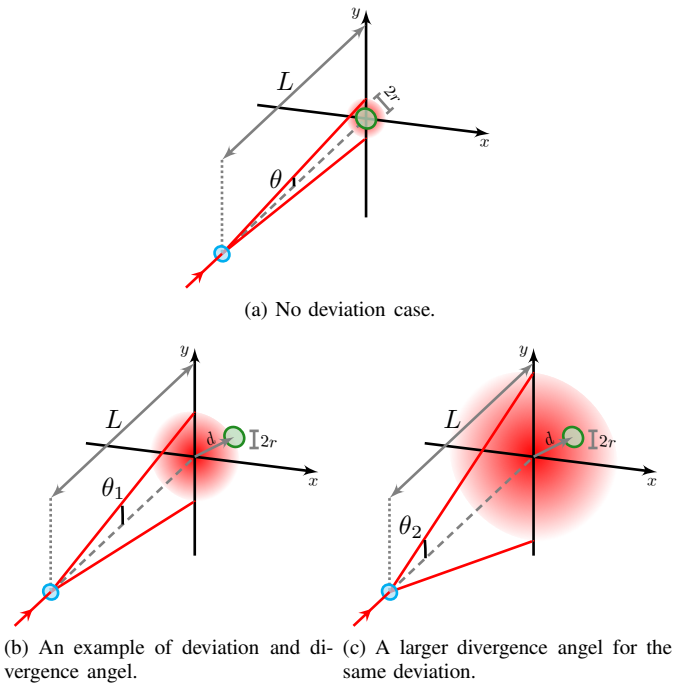
\begin{figure}[t]
    \tikzset{every picture/.style={line width=1pt}} 
    \begin{center}
        \subfloat[No deviation case.]{
            \label{fig:system_model-a}
            \begin{tikzpicture}[scale=0.7]
                \begin{scope}[tdplot_main_coords]
                    \begin{scope}[canvas is yz plane at x=8,xscale=1,yscale=1,transform shape]
                        \draw[inner color=red, draw=none, opacity=1] (0,0) circle (0.4);
                    \end{scope}
        
                    \draw[gray,densely dashed] (0,0,0) -- (8,0,0);
        
                    \begin{scope}[canvas is yz plane at x=8,xscale=1,yscale=1,transform shape]
                    
                        \draw[-latex'] (-2.2,0) -- (2.5,0);
                        \draw[-latex'] (0,-2) -- (0,2.2);
                        \node[] at (2.4,-0.2) {$x$};
                        \node[] at (-0.2,2.1) {$y$};
        
                        \node[circle,fill,red,inner sep=0pt] (A) at ($(0,0)+(90:0.4)$) {};
                        \node[circle,fill,red,inner sep=0pt] (B) at ($(0,0)+(270:0.4)$) {};
                        \filldraw[fill=ForestGreen!30,draw=ForestGreen,fill opacity=0.8] (0, 0) circle (0.2);
                        \draw[{Bar[scale width=0.5]}-{Bar[scale width=0.5]}, gray] (0.5, 0.2) -- (0.2, 0.5);
                        \node[rotate=315] at (0.55, 0.55) {$2r$};
                    \end{scope}
                    
                    \draw[red] (A) -- (0,0,0) (B) -- (0,0,0);
        
                    \begin{scope}[canvas is xz plane at y=0]
                        \node[rotate=0] at (3,0.5) {$\theta$};
                        \draw (4,0) arc (0:5:2);
                        \draw [latex'-latex',gray] (0,1.9) -- (8,1.9);
                        \draw [densely dotted,gray] (0,1.9) -- (0,0);
                        \node[rotate=0] at (3.5, 2.3) {$L$};
                    \end{scope}
        
                    \begin{scope}[canvas is zy plane at x=0]
                        \filldraw[fill=cyan!30,draw=cyan,fill opacity=0.7] (0,0) circle (0.15);
                    \end{scope}
        
                    \draw [red,postaction={decorate,decoration={markings,mark=at position 0.5
                    with {\arrow{stealth};}}}] (-2,0,0) -- (0,0,0);
                    
                \end{scope}
            \end{tikzpicture}
        }
    \end{center}
    
    \subfloat[An example of deviation and divergence angel.]{
        \label{fig:system_model-b}
        \begin{tikzpicture}[scale=0.7]
            \begin{scope}[tdplot_main_coords]
                \begin{scope}[canvas is yz plane at x=8,xscale=1,yscale=1,transform shape]
                    \draw[inner color=red, draw=none, opacity=1] (0,0) circle (0.9);
                \end{scope}
    
                \draw[gray,densely dashed] (0,0,0) -- (8,0,0);
    
                \begin{scope}[canvas is yz plane at x=8,xscale=1,yscale=1,transform shape]
                
                    \draw[-latex'] (-2.2,0) -- (2.5,0);
                    \draw[-latex'] (0,-2) -- (0,2.2);
                    \node[] at (2.4,-0.2) {$x$};
                    \node[] at (-0.2,2.1) {$y$};
    
                    \node[circle,fill,red,inner sep=0pt] (A) at ($(0,0)+(90:0.9)$) {};
                    \node[circle,fill,red,inner sep=0pt] (B) at ($(0,0)+(270:0.9)$) {};
                    \filldraw[fill=ForestGreen!30,draw=ForestGreen,fill opacity=0.8] (0.8,0.5) circle (0.2);
                    \draw[-latex',gray] (0,0) -- (0.8,0.5);
                    \draw[{Bar[scale width=0.5]}-{Bar[scale width=0.5]}, gray] (1.2, 0.3) -- (1.2, 0.7);
                    \node[rotate=0] at (1.5, 0.5) {$2r$};
                    \node[rotate=30] at (0.3,0.4) {$d$};
                \end{scope}
    
                \draw[red] (A) -- (0,0,0) (B) -- (0,0,0);
    
                \begin{scope}[canvas is xz plane at y=0]
                    \node[rotate=0] at (3.5, 0.9) {$\theta_1$};
                    \draw (4,0) arc (0:5.5:4);
                    \draw [latex'-latex',gray] (0,1.9) -- (8,1.9);
                    \draw [densely dotted,gray] (0,1.9) -- (0,0);
                    \node[rotate=0] at (3.5, 2.3) {$L$};
                \end{scope}
    
                \begin{scope}[canvas is zy plane at x=0]
                    \filldraw[fill=cyan!30,draw=cyan,fill opacity=0.7] (0,0) circle (0.15);
                \end{scope}
    
                \draw [red,postaction={decorate,decoration={markings,mark=at position 0.5
                with {\arrow{stealth};}}}] (-2,0,0) -- (0,0,0);
                
            \end{scope}
        \end{tikzpicture}
    }
    \
    \subfloat[A larger divergence angel for the same deviation.]{
        \label{fig:system_model-c}
        \begin{tikzpicture}[scale=0.7]
            \begin{scope}[tdplot_main_coords]
                \begin{scope}[canvas is yz plane at x=8,xscale=1,yscale=1,transform shape]
                    \draw[inner color=red, draw=none, opacity=0.75] (0,0) circle (1.7);
                \end{scope}
    
                \draw[gray,densely dashed] (0,0,0) -- (8,0,0);
    
                \begin{scope}[canvas is yz plane at x=8,xscale=1,yscale=1,transform shape]
                
                    \draw[-latex'] (-2.2,0) -- (2.5,0);
                    \draw[-latex'] (0,-2) -- (0,2.2);
                    \node[] at (2.4,-0.2) {$x$};
                    \node[] at (-0.2,2.1) {$y$};
    
                    \node[circle,fill,red,inner sep=0pt] (A) at ($(\X,0)+(90:\Y)$) {};
                    \node[circle,fill,red,inner sep=0pt] (B) at ($(-\X,0)+(270:\Y)$) {};
                    \filldraw[fill=ForestGreen!30,draw=ForestGreen,fill opacity=0.8] (0.8,0.5) circle (0.2);
                    \draw[-latex',gray] (0,0) -- (0.8,0.5);
                    \draw[{Bar[scale width=0.5]}-{Bar[scale width=0.5]}, gray] (1.2, 0.3) -- (1.2, 0.7);
                    \node[rotate=0] at (1.5, 0.5) {$2r$};
                    \node[rotate=30] at (0.3,0.4) {$d$};
                \end{scope}
    
                \draw[red] (A) -- (0,0,0) (B) -- (0,0,0);
    
                \begin{scope}[canvas is xz plane at y=0]
                    \node[rotate=0] at (1.2,0.9) {$\theta_2$};
                    \draw (2,0) arc (0:11:2);
                    \draw [latex'-latex',gray] (0,1.9) -- (8,1.9);
                    \draw [densely dotted,gray] (0,1.9) -- (0,0);
                    \node[rotate=0] at (3.5, 2.3) {$L$};
                \end{scope}
    
                \begin{scope}[canvas is zy plane at x=0]
                    \filldraw[fill=cyan!30,draw=cyan,fill opacity=0.7] (0,0) circle (0.15);
                \end{scope}
    
                \draw [red,postaction={decorate,decoration={markings,mark=at position 0.5
                with {\arrow{stealth};}}}] (-2,0,0) -- (0,0,0);
                
            \end{scope}
        \end{tikzpicture}
        }

    \caption{The figures depict the transmission of a laser beam to the receiver detection plane. Figures \ref{fig:system_model-b} and \ref{fig:system_model-c} show the same deviation but tow different different divergence angles ($\theta_1 < \theta_2$).}\label{fig:system_model}
\end{figure}

%% file: fig-shortest-path.tex
\tikzoption{canvas is xy plane at z}[]{%
  \def\tikz@plane@origin{\pgfpointxyz{0}{0}{#1}}%
  \def\tikz@plane@x{\pgfpointxyz{1}{0}{#1}}%
  \def\tikz@plane@y{\pgfpointxyz{0}{1}{#1}}%
  \tikz@canvas@is@plane
}

\definecolor{ForestGreen}{RGB}{34,139,34}

\pgfmathsetmacro{\X}{0}
\pgfmathsetmacro{\Y}{1.2}
\tdplotsetmaincoords{110}{70}

\begin{figure*}[t]
    \begin{center}

\tikzset{every picture/.style={line width=0.75pt}} 

\begin{tikzpicture}[x=0.75pt,y=0.75pt,yscale=-1,xscale=1]

\draw    (372.5,1527.33) -- (406.04,1543.7) ;
\draw [shift={(407.83,1544.58)}, rotate = 206.02] [fill={rgb, 255:red, 0; green, 0; blue, 0 }  ][line width=0.08]  [draw opacity=0] (7.2,-1.8) -- (0,0) -- (7.2,1.8) -- cycle    ;
\draw    (368.5,1537.33) -- (405.06,1546.27) ;
\draw [shift={(407,1546.75)}, rotate = 193.74] [fill={rgb, 255:red, 0; green, 0; blue, 0 }  ][line width=0.08]  [draw opacity=0] (7.2,-1.8) -- (0,0) -- (7.2,1.8) -- cycle    ;
\draw    (371,1549.08) -- (404.5,1549.08) ;
\draw [shift={(406.5,1549.08)}, rotate = 180] [fill={rgb, 255:red, 0; green, 0; blue, 0 }  ][line width=0.08]  [draw opacity=0] (7.2,-1.8) -- (0,0) -- (7.2,1.8) -- cycle    ;
\draw    (367.25,1561.33) -- (404.56,1552.06) ;
\draw [shift={(406.5,1551.58)}, rotate = 166.05] [fill={rgb, 255:red, 0; green, 0; blue, 0 }  ][line width=0.08]  [draw opacity=0] (7.2,-1.8) -- (0,0) -- (7.2,1.8) -- cycle    ;
\draw    (371.5,1571.33) -- (405.7,1554.79) ;
\draw [shift={(407.5,1553.91)}, rotate = 154.18] [fill={rgb, 255:red, 0; green, 0; blue, 0 }  ][line width=0.08]  [draw opacity=0] (7.2,-1.8) -- (0,0) -- (7.2,1.8) -- cycle    ;
\draw    (372.75,1557.58) -- (406.29,1573.95) ;
\draw [shift={(408.08,1574.83)}, rotate = 206.02] [fill={rgb, 255:red, 0; green, 0; blue, 0 }  ][line width=0.08]  [draw opacity=0] (7.2,-1.8) -- (0,0) -- (7.2,1.8) -- cycle    ;
\draw    (368.75,1567.58) -- (405.31,1576.52) ;
\draw [shift={(407.25,1577)}, rotate = 193.74] [fill={rgb, 255:red, 0; green, 0; blue, 0 }  ][line width=0.08]  [draw opacity=0] (7.2,-1.8) -- (0,0) -- (7.2,1.8) -- cycle    ;
\draw    (371.25,1579.33) -- (404.75,1579.33) ;
\draw [shift={(406.75,1579.33)}, rotate = 180] [fill={rgb, 255:red, 0; green, 0; blue, 0 }  ][line width=0.08]  [draw opacity=0] (7.2,-1.8) -- (0,0) -- (7.2,1.8) -- cycle    ;
\draw    (367.5,1591.58) -- (404.81,1582.31) ;
\draw [shift={(406.75,1581.83)}, rotate = 166.05] [fill={rgb, 255:red, 0; green, 0; blue, 0 }  ][line width=0.08]  [draw opacity=0] (7.2,-1.8) -- (0,0) -- (7.2,1.8) -- cycle    ;
\draw    (371.75,1601.58) -- (405.95,1585.04) ;
\draw [shift={(407.75,1584.16)}, rotate = 154.18] [fill={rgb, 255:red, 0; green, 0; blue, 0 }  ][line width=0.08]  [draw opacity=0] (7.2,-1.8) -- (0,0) -- (7.2,1.8) -- cycle    ;
\draw    (372.5,1587.58) -- (406.04,1603.95) ;
\draw [shift={(407.83,1604.83)}, rotate = 206.02] [fill={rgb, 255:red, 0; green, 0; blue, 0 }  ][line width=0.08]  [draw opacity=0] (7.2,-1.8) -- (0,0) -- (7.2,1.8) -- cycle    ;
\draw    (368.5,1597.58) -- (405.06,1606.52) ;
\draw [shift={(407,1607)}, rotate = 193.74] [fill={rgb, 255:red, 0; green, 0; blue, 0 }  ][line width=0.08]  [draw opacity=0] (7.2,-1.8) -- (0,0) -- (7.2,1.8) -- cycle    ;
\draw    (371,1609.33) -- (404.5,1609.33) ;
\draw [shift={(406.5,1609.33)}, rotate = 180] [fill={rgb, 255:red, 0; green, 0; blue, 0 }  ][line width=0.08]  [draw opacity=0] (7.2,-1.8) -- (0,0) -- (7.2,1.8) -- cycle    ;
\draw    (367.25,1621.58) -- (404.56,1612.31) ;
\draw [shift={(406.5,1611.83)}, rotate = 166.05] [fill={rgb, 255:red, 0; green, 0; blue, 0 }  ][line width=0.08]  [draw opacity=0] (7.2,-1.8) -- (0,0) -- (7.2,1.8) -- cycle    ;
\draw    (371.5,1631.58) -- (405.7,1615.04) ;
\draw [shift={(407.5,1614.16)}, rotate = 154.18] [fill={rgb, 255:red, 0; green, 0; blue, 0 }  ][line width=0.08]  [draw opacity=0] (7.2,-1.8) -- (0,0) -- (7.2,1.8) -- cycle    ;
\draw    (372.75,1618.58) -- (406.29,1634.95) ;
\draw [shift={(408.08,1635.83)}, rotate = 206.02] [fill={rgb, 255:red, 0; green, 0; blue, 0 }  ][line width=0.08]  [draw opacity=0] (7.2,-1.8) -- (0,0) -- (7.2,1.8) -- cycle    ;
\draw    (368.75,1628.58) -- (405.31,1637.52) ;
\draw [shift={(407.25,1638)}, rotate = 193.74] [fill={rgb, 255:red, 0; green, 0; blue, 0 }  ][line width=0.08]  [draw opacity=0] (7.2,-1.8) -- (0,0) -- (7.2,1.8) -- cycle    ;
\draw    (371.25,1640.33) -- (404.75,1640.33) ;
\draw [shift={(406.75,1640.33)}, rotate = 180] [fill={rgb, 255:red, 0; green, 0; blue, 0 }  ][line width=0.08]  [draw opacity=0] (7.2,-1.8) -- (0,0) -- (7.2,1.8) -- cycle    ;
\draw    (367.5,1652.58) -- (404.81,1643.31) ;
\draw [shift={(406.75,1642.83)}, rotate = 166.05] [fill={rgb, 255:red, 0; green, 0; blue, 0 }  ][line width=0.08]  [draw opacity=0] (7.2,-1.8) -- (0,0) -- (7.2,1.8) -- cycle    ;
\draw    (371.75,1662.58) -- (405.95,1646.04) ;
\draw [shift={(407.75,1645.16)}, rotate = 154.18] [fill={rgb, 255:red, 0; green, 0; blue, 0 }  ][line width=0.08]  [draw opacity=0] (7.2,-1.8) -- (0,0) -- (7.2,1.8) -- cycle    ;
\draw    (371,1518.33) -- (404.5,1518.33) ;
\draw [shift={(406.5,1518.33)}, rotate = 180] [fill={rgb, 255:red, 0; green, 0; blue, 0 }  ][line width=0.08]  [draw opacity=0] (7.2,-1.8) -- (0,0) -- (7.2,1.8) -- cycle    ;
\draw    (367.25,1530.58) -- (404.56,1521.31) ;
\draw [shift={(406.5,1520.83)}, rotate = 166.05] [fill={rgb, 255:red, 0; green, 0; blue, 0 }  ][line width=0.08]  [draw opacity=0] (7.2,-1.8) -- (0,0) -- (7.2,1.8) -- cycle    ;
\draw    (371.5,1540.58) -- (405.7,1524.04) ;
\draw [shift={(407.5,1523.16)}, rotate = 154.18] [fill={rgb, 255:red, 0; green, 0; blue, 0 }  ][line width=0.08]  [draw opacity=0] (7.2,-1.8) -- (0,0) -- (7.2,1.8) -- cycle    ;
\draw  [color={rgb, 255:red, 126; green, 211; blue, 33 }  ,draw opacity=1 ][fill={rgb, 255:red, 255; green, 255; blue, 255 }  ,fill opacity=1 ] (492.8,1518.82) .. controls (492.8,1513.57) and (497.05,1509.32) .. (502.3,1509.32) .. controls (507.55,1509.32) and (511.8,1513.57) .. (511.8,1518.82) .. controls (511.8,1524.06) and (507.55,1528.32) .. (502.3,1528.32) .. controls (497.05,1528.32) and (492.8,1524.06) .. (492.8,1518.82) -- cycle ;
\draw  [color={rgb, 255:red, 126; green, 211; blue, 33 }  ,draw opacity=1 ][fill={rgb, 255:red, 255; green, 255; blue, 255 }  ,fill opacity=1 ] (50,1518.83) .. controls (50,1513.59) and (54.25,1509.33) .. (59.5,1509.33) .. controls (64.75,1509.33) and (69,1513.59) .. (69,1518.83) .. controls (69,1524.08) and (64.75,1528.33) .. (59.5,1528.33) .. controls (54.25,1528.33) and (50,1524.08) .. (50,1518.83) -- cycle ;
\draw  [color={rgb, 255:red, 126; green, 211; blue, 33 }  ,draw opacity=1 ][fill={rgb, 255:red, 255; green, 255; blue, 255 }  ,fill opacity=1 ][dash pattern={on 0.84pt off 2.51pt}] (50,1549.17) .. controls (50,1543.92) and (54.25,1539.67) .. (59.5,1539.67) .. controls (64.75,1539.67) and (69,1543.92) .. (69,1549.17) .. controls (69,1554.41) and (64.75,1558.67) .. (59.5,1558.67) .. controls (54.25,1558.67) and (50,1554.41) .. (50,1549.17) -- cycle ;
\draw  [color={rgb, 255:red, 126; green, 211; blue, 33 }  ,draw opacity=1 ][fill={rgb, 255:red, 255; green, 255; blue, 255 }  ,fill opacity=1 ][dash pattern={on 0.84pt off 2.51pt}] (50,1579.5) .. controls (50,1574.25) and (54.25,1570) .. (59.5,1570) .. controls (64.75,1570) and (69,1574.25) .. (69,1579.5) .. controls (69,1584.75) and (64.75,1589) .. (59.5,1589) .. controls (54.25,1589) and (50,1584.75) .. (50,1579.5) -- cycle ;
\draw  [color={rgb, 255:red, 126; green, 211; blue, 33 }  ,draw opacity=1 ][fill={rgb, 255:red, 255; green, 255; blue, 255 }  ,fill opacity=1 ][dash pattern={on 0.84pt off 2.51pt}] (50,1609.83) .. controls (50,1604.59) and (54.25,1600.33) .. (59.5,1600.33) .. controls (64.75,1600.33) and (69,1604.59) .. (69,1609.83) .. controls (69,1615.08) and (64.75,1619.33) .. (59.5,1619.33) .. controls (54.25,1619.33) and (50,1615.08) .. (50,1609.83) -- cycle ;
\draw  [color={rgb, 255:red, 126; green, 211; blue, 33 }  ,draw opacity=1 ][fill={rgb, 255:red, 255; green, 255; blue, 255 }  ,fill opacity=1 ] (160,1518.83) .. controls (160,1513.59) and (164.25,1509.33) .. (169.5,1509.33) .. controls (174.75,1509.33) and (179,1513.59) .. (179,1518.83) .. controls (179,1524.08) and (174.75,1528.33) .. (169.5,1528.33) .. controls (164.25,1528.33) and (160,1524.08) .. (160,1518.83) -- cycle ;
\draw  [color={rgb, 255:red, 126; green, 211; blue, 33 }  ,draw opacity=1 ][fill={rgb, 255:red, 255; green, 255; blue, 255 }  ,fill opacity=1 ] (160,1549.17) .. controls (160,1543.92) and (164.25,1539.67) .. (169.5,1539.67) .. controls (174.75,1539.67) and (179,1543.92) .. (179,1549.17) .. controls (179,1554.41) and (174.75,1558.67) .. (169.5,1558.67) .. controls (164.25,1558.67) and (160,1554.41) .. (160,1549.17) -- cycle ;
\draw  [color={rgb, 255:red, 126; green, 211; blue, 33 }  ,draw opacity=1 ][fill={rgb, 255:red, 255; green, 255; blue, 255 }  ,fill opacity=1 ] (160,1579.5) .. controls (160,1574.25) and (164.25,1570) .. (169.5,1570) .. controls (174.75,1570) and (179,1574.25) .. (179,1579.5) .. controls (179,1584.75) and (174.75,1589) .. (169.5,1589) .. controls (164.25,1589) and (160,1584.75) .. (160,1579.5) -- cycle ;
\draw  [color={rgb, 255:red, 126; green, 211; blue, 33 }  ,draw opacity=1 ][fill={rgb, 255:red, 255; green, 255; blue, 255 }  ,fill opacity=1 ][dash pattern={on 0.84pt off 2.51pt}] (160,1609.83) .. controls (160,1604.59) and (164.25,1600.33) .. (169.5,1600.33) .. controls (174.75,1600.33) and (179,1604.59) .. (179,1609.83) .. controls (179,1615.08) and (174.75,1619.33) .. (169.5,1619.33) .. controls (164.25,1619.33) and (160,1615.08) .. (160,1609.83) -- cycle ;
\draw  [color={rgb, 255:red, 126; green, 211; blue, 33 }  ,draw opacity=1 ][fill={rgb, 255:red, 255; green, 255; blue, 255 }  ,fill opacity=1 ] (270,1518.83) .. controls (270,1513.59) and (274.25,1509.33) .. (279.5,1509.33) .. controls (284.75,1509.33) and (289,1513.59) .. (289,1518.83) .. controls (289,1524.08) and (284.75,1528.33) .. (279.5,1528.33) .. controls (274.25,1528.33) and (270,1524.08) .. (270,1518.83) -- cycle ;
\draw  [color={rgb, 255:red, 126; green, 211; blue, 33 }  ,draw opacity=1 ][fill={rgb, 255:red, 255; green, 255; blue, 255 }  ,fill opacity=1 ] (270,1549.17) .. controls (270,1543.92) and (274.25,1539.67) .. (279.5,1539.67) .. controls (284.75,1539.67) and (289,1543.92) .. (289,1549.17) .. controls (289,1554.41) and (284.75,1558.67) .. (279.5,1558.67) .. controls (274.25,1558.67) and (270,1554.41) .. (270,1549.17) -- cycle ;
\draw  [color={rgb, 255:red, 126; green, 211; blue, 33 }  ,draw opacity=1 ][fill={rgb, 255:red, 255; green, 255; blue, 255 }  ,fill opacity=1 ] (270,1579.5) .. controls (270,1574.25) and (274.25,1570) .. (279.5,1570) .. controls (284.75,1570) and (289,1574.25) .. (289,1579.5) .. controls (289,1584.75) and (284.75,1589) .. (279.5,1589) .. controls (274.25,1589) and (270,1584.75) .. (270,1579.5) -- cycle ;
\draw  [color={rgb, 255:red, 126; green, 211; blue, 33 }  ,draw opacity=1 ][fill={rgb, 255:red, 255; green, 255; blue, 255 }  ,fill opacity=1 ] (270,1609.83) .. controls (270,1604.59) and (274.25,1600.33) .. (279.5,1600.33) .. controls (284.75,1600.33) and (289,1604.59) .. (289,1609.83) .. controls (289,1615.08) and (284.75,1619.33) .. (279.5,1619.33) .. controls (274.25,1619.33) and (270,1615.08) .. (270,1609.83) -- cycle ;
\draw  [color={rgb, 255:red, 126; green, 211; blue, 33 }  ,draw opacity=1 ][fill={rgb, 255:red, 255; green, 255; blue, 255 }  ,fill opacity=1 ] (407,1518.83) .. controls (407,1513.59) and (411.25,1509.33) .. (416.5,1509.33) .. controls (421.75,1509.33) and (426,1513.59) .. (426,1518.83) .. controls (426,1524.08) and (421.75,1528.33) .. (416.5,1528.33) .. controls (411.25,1528.33) and (407,1524.08) .. (407,1518.83) -- cycle ;
\draw  [color={rgb, 255:red, 126; green, 211; blue, 33 }  ,draw opacity=1 ][fill={rgb, 255:red, 255; green, 255; blue, 255 }  ,fill opacity=1 ] (407,1549.17) .. controls (407,1543.92) and (411.25,1539.67) .. (416.5,1539.67) .. controls (421.75,1539.67) and (426,1543.92) .. (426,1549.17) .. controls (426,1554.41) and (421.75,1558.67) .. (416.5,1558.67) .. controls (411.25,1558.67) and (407,1554.41) .. (407,1549.17) -- cycle ;
\draw  [color={rgb, 255:red, 126; green, 211; blue, 33 }  ,draw opacity=1 ][fill={rgb, 255:red, 255; green, 255; blue, 255 }  ,fill opacity=1 ] (407,1579.5) .. controls (407,1574.25) and (411.25,1570) .. (416.5,1570) .. controls (421.75,1570) and (426,1574.25) .. (426,1579.5) .. controls (426,1584.75) and (421.75,1589) .. (416.5,1589) .. controls (411.25,1589) and (407,1584.75) .. (407,1579.5) -- cycle ;
\draw  [color={rgb, 255:red, 126; green, 211; blue, 33 }  ,draw opacity=1 ][fill={rgb, 255:red, 255; green, 255; blue, 255 }  ,fill opacity=1 ] (407,1609.83) .. controls (407,1604.59) and (411.25,1600.33) .. (416.5,1600.33) .. controls (421.75,1600.33) and (426,1604.59) .. (426,1609.83) .. controls (426,1615.08) and (421.75,1619.33) .. (416.5,1619.33) .. controls (411.25,1619.33) and (407,1615.08) .. (407,1609.83) -- cycle ;
\draw  [color={rgb, 255:red, 126; green, 211; blue, 33 }  ,draw opacity=1 ][fill={rgb, 255:red, 255; green, 255; blue, 255 }  ,fill opacity=1 ][dash pattern={on 0.84pt off 2.51pt}] (50,1639.83) .. controls (50,1634.59) and (54.25,1630.33) .. (59.5,1630.33) .. controls (64.75,1630.33) and (69,1634.59) .. (69,1639.83) .. controls (69,1645.08) and (64.75,1649.33) .. (59.5,1649.33) .. controls (54.25,1649.33) and (50,1645.08) .. (50,1639.83) -- cycle ;
\draw  [color={rgb, 255:red, 126; green, 211; blue, 33 }  ,draw opacity=1 ][fill={rgb, 255:red, 255; green, 255; blue, 255 }  ,fill opacity=1 ][dash pattern={on 0.84pt off 2.51pt}] (160,1640.33) .. controls (160,1635.09) and (164.25,1630.83) .. (169.5,1630.83) .. controls (174.75,1630.83) and (179,1635.09) .. (179,1640.33) .. controls (179,1645.58) and (174.75,1649.83) .. (169.5,1649.83) .. controls (164.25,1649.83) and (160,1645.58) .. (160,1640.33) -- cycle ;
\draw  [color={rgb, 255:red, 126; green, 211; blue, 33 }  ,draw opacity=1 ][fill={rgb, 255:red, 255; green, 255; blue, 255 }  ,fill opacity=1 ] (270,1640.33) .. controls (270,1635.09) and (274.25,1630.83) .. (279.5,1630.83) .. controls (284.75,1630.83) and (289,1635.09) .. (289,1640.33) .. controls (289,1645.58) and (284.75,1649.83) .. (279.5,1649.83) .. controls (274.25,1649.83) and (270,1645.58) .. (270,1640.33) -- cycle ;
\draw  [color={rgb, 255:red, 126; green, 211; blue, 33 }  ,draw opacity=1 ][fill={rgb, 255:red, 255; green, 255; blue, 255 }  ,fill opacity=1 ] (407,1640.33) .. controls (407,1635.09) and (411.25,1630.83) .. (416.5,1630.83) .. controls (421.75,1630.83) and (426,1635.09) .. (426,1640.33) .. controls (426,1645.58) and (421.75,1649.83) .. (416.5,1649.83) .. controls (411.25,1649.83) and (407,1645.58) .. (407,1640.33) -- cycle ;
\draw    (68.91,1518.83) -- (158.08,1518.83) ;
\draw [shift={(160.08,1518.83)}, rotate = 180] [fill={rgb, 255:red, 0; green, 0; blue, 0 }  ][line width=0.08]  [draw opacity=0] (7.2,-1.8) -- (0,0) -- (7.2,1.8) -- cycle    ;
\draw    (68.75,1522.17) -- (158.3,1546.31) ;
\draw [shift={(160.23,1546.83)}, rotate = 195.09] [fill={rgb, 255:red, 0; green, 0; blue, 0 }  ][line width=0.08]  [draw opacity=0] (7.2,-1.8) -- (0,0) -- (7.2,1.8) -- cycle    ;
\draw    (67.25,1524.83) -- (158.78,1575.2) ;
\draw [shift={(160.53,1576.16)}, rotate = 208.82] [fill={rgb, 255:red, 0; green, 0; blue, 0 }  ][line width=0.08]  [draw opacity=0] (7.2,-1.8) -- (0,0) -- (7.2,1.8) -- cycle    ;
\draw    (178.41,1518.59) -- (267.79,1518.59) ;
\draw [shift={(269.79,1518.59)}, rotate = 180] [fill={rgb, 255:red, 0; green, 0; blue, 0 }  ][line width=0.08]  [draw opacity=0] (7.2,-1.8) -- (0,0) -- (7.2,1.8) -- cycle    ;
\draw    (178.26,1521.91) -- (268.01,1546.04) ;
\draw [shift={(269.94,1546.56)}, rotate = 195.05] [fill={rgb, 255:red, 0; green, 0; blue, 0 }  ][line width=0.08]  [draw opacity=0] (7.2,-1.8) -- (0,0) -- (7.2,1.8) -- cycle    ;
\draw    (176.75,1524.58) -- (268.93,1573.77) ;
\draw [shift={(270.7,1574.71)}, rotate = 208.08] [fill={rgb, 255:red, 0; green, 0; blue, 0 }  ][line width=0.08]  [draw opacity=0] (7.2,-1.8) -- (0,0) -- (7.2,1.8) -- cycle    ;
\draw    (178.71,1549.06) -- (268.1,1549.06) ;
\draw [shift={(270.1,1549.06)}, rotate = 180] [fill={rgb, 255:red, 0; green, 0; blue, 0 }  ][line width=0.08]  [draw opacity=0] (7.2,-1.8) -- (0,0) -- (7.2,1.8) -- cycle    ;
\draw    (177.96,1546.06) -- (267.86,1521.61) ;
\draw [shift={(269.79,1521.08)}, rotate = 164.78] [fill={rgb, 255:red, 0; green, 0; blue, 0 }  ][line width=0.08]  [draw opacity=0] (7.2,-1.8) -- (0,0) -- (7.2,1.8) -- cycle    ;
\draw    (178.56,1552.39) -- (268.31,1576.52) ;
\draw [shift={(270.25,1577.04)}, rotate = 195.05] [fill={rgb, 255:red, 0; green, 0; blue, 0 }  ][line width=0.08]  [draw opacity=0] (7.2,-1.8) -- (0,0) -- (7.2,1.8) -- cycle    ;
\draw    (177.05,1555.06) -- (268.94,1604.41) ;
\draw [shift={(270.7,1605.35)}, rotate = 208.24] [fill={rgb, 255:red, 0; green, 0; blue, 0 }  ][line width=0.08]  [draw opacity=0] (7.2,-1.8) -- (0,0) -- (7.2,1.8) -- cycle    ;
\draw    (179.01,1579.37) -- (268.4,1579.37) ;
\draw [shift={(270.4,1579.37)}, rotate = 180] [fill={rgb, 255:red, 0; green, 0; blue, 0 }  ][line width=0.08]  [draw opacity=0] (7.2,-1.8) -- (0,0) -- (7.2,1.8) -- cycle    ;
\draw    (178.26,1576.37) -- (268.02,1551.59) ;
\draw [shift={(269.94,1551.06)}, rotate = 164.56] [fill={rgb, 255:red, 0; green, 0; blue, 0 }  ][line width=0.08]  [draw opacity=0] (7.2,-1.8) -- (0,0) -- (7.2,1.8) -- cycle    ;
\draw    (177.35,1573.38) -- (268.94,1524.36) ;
\draw [shift={(270.7,1523.41)}, rotate = 151.84] [fill={rgb, 255:red, 0; green, 0; blue, 0 }  ][line width=0.08]  [draw opacity=0] (7.2,-1.8) -- (0,0) -- (7.2,1.8) -- cycle    ;
\draw    (178.86,1582.7) -- (268.17,1607.16) ;
\draw [shift={(270.1,1607.68)}, rotate = 195.31] [fill={rgb, 255:red, 0; green, 0; blue, 0 }  ][line width=0.08]  [draw opacity=0] (7.2,-1.8) -- (0,0) -- (7.2,1.8) -- cycle    ;
\draw    (177.35,1585.37) -- (269.1,1635.7) ;
\draw [shift={(270.85,1636.66)}, rotate = 208.75] [fill={rgb, 255:red, 0; green, 0; blue, 0 }  ][line width=0.08]  [draw opacity=0] (7.2,-1.8) -- (0,0) -- (7.2,1.8) -- cycle    ;
\draw    (289,1518.5) -- (353,1518.52) ;
\draw [shift={(355,1518.52)}, rotate = 180.01] [fill={rgb, 255:red, 0; green, 0; blue, 0 }  ][line width=0.08]  [draw opacity=0] (7.2,-1.8) -- (0,0) -- (7.2,1.8) -- cycle    ;
\draw    (288.83,1521.83) -- (351.05,1536.56) ;
\draw [shift={(353,1537.02)}, rotate = 193.31] [fill={rgb, 255:red, 0; green, 0; blue, 0 }  ][line width=0.08]  [draw opacity=0] (7.2,-1.8) -- (0,0) -- (7.2,1.8) -- cycle    ;
\draw    (287.17,1524.5) -- (349.69,1553.67) ;
\draw [shift={(351.5,1554.52)}, rotate = 205.01] [fill={rgb, 255:red, 0; green, 0; blue, 0 }  ][line width=0.08]  [draw opacity=0] (7.2,-1.8) -- (0,0) -- (7.2,1.8) -- cycle    ;
\draw    (288.67,1548.83) -- (351.5,1548.83) ;
\draw [shift={(353.5,1548.83)}, rotate = 180] [fill={rgb, 255:red, 0; green, 0; blue, 0 }  ][line width=0.08]  [draw opacity=0] (7.2,-1.8) -- (0,0) -- (7.2,1.8) -- cycle    ;
\draw    (287.83,1545.83) -- (351.06,1530) ;
\draw [shift={(353,1529.52)}, rotate = 165.94] [fill={rgb, 255:red, 0; green, 0; blue, 0 }  ][line width=0.08]  [draw opacity=0] (7.2,-1.8) -- (0,0) -- (7.2,1.8) -- cycle    ;
\draw    (288.5,1552.16) -- (352.56,1567.55) ;
\draw [shift={(354.5,1568.02)}, rotate = 193.51] [fill={rgb, 255:red, 0; green, 0; blue, 0 }  ][line width=0.08]  [draw opacity=0] (7.2,-1.8) -- (0,0) -- (7.2,1.8) -- cycle    ;
\draw    (286.83,1554.83) -- (348.69,1583.67) ;
\draw [shift={(350.5,1584.52)}, rotate = 205] [fill={rgb, 255:red, 0; green, 0; blue, 0 }  ][line width=0.08]  [draw opacity=0] (7.2,-1.8) -- (0,0) -- (7.2,1.8) -- cycle    ;
\draw    (289,1579.17) -- (350,1579.17) ;
\draw [shift={(352,1579.17)}, rotate = 180] [fill={rgb, 255:red, 0; green, 0; blue, 0 }  ][line width=0.08]  [draw opacity=0] (7.2,-1.8) -- (0,0) -- (7.2,1.8) -- cycle    ;
\draw    (288.17,1576.16) -- (351.05,1561.47) ;
\draw [shift={(353,1561.02)}, rotate = 166.85] [fill={rgb, 255:red, 0; green, 0; blue, 0 }  ][line width=0.08]  [draw opacity=0] (7.2,-1.8) -- (0,0) -- (7.2,1.8) -- cycle    ;
\draw    (287.17,1573.16) -- (350.21,1541.9) ;
\draw [shift={(352,1541.02)}, rotate = 153.62] [fill={rgb, 255:red, 0; green, 0; blue, 0 }  ][line width=0.08]  [draw opacity=0] (7.2,-1.8) -- (0,0) -- (7.2,1.8) -- cycle    ;
\draw    (288.83,1582.5) -- (347.55,1596.55) ;
\draw [shift={(349.5,1597.02)}, rotate = 193.46] [fill={rgb, 255:red, 0; green, 0; blue, 0 }  ][line width=0.08]  [draw opacity=0] (7.2,-1.8) -- (0,0) -- (7.2,1.8) -- cycle    ;
\draw    (287.17,1585.16) -- (348.7,1614.65) ;
\draw [shift={(350.5,1615.52)}, rotate = 205.61] [fill={rgb, 255:red, 0; green, 0; blue, 0 }  ][line width=0.08]  [draw opacity=0] (7.2,-1.8) -- (0,0) -- (7.2,1.8) -- cycle    ;
\draw    (289,1609.5) -- (348.5,1609.5) ;
\draw [shift={(350.5,1609.5)}, rotate = 180] [fill={rgb, 255:red, 0; green, 0; blue, 0 }  ][line width=0.08]  [draw opacity=0] (7.2,-1.8) -- (0,0) -- (7.2,1.8) -- cycle    ;
\draw    (288.17,1606.5) -- (346.56,1592.48) ;
\draw [shift={(348.5,1592.02)}, rotate = 166.5] [fill={rgb, 255:red, 0; green, 0; blue, 0 }  ][line width=0.08]  [draw opacity=0] (7.2,-1.8) -- (0,0) -- (7.2,1.8) -- cycle    ;
\draw    (287.17,1603.5) -- (347.61,1573.31) ;
\draw [shift={(349.4,1572.42)}, rotate = 153.46] [fill={rgb, 255:red, 0; green, 0; blue, 0 }  ][line width=0.08]  [draw opacity=0] (7.2,-1.8) -- (0,0) -- (7.2,1.8) -- cycle    ;
\draw    (288.83,1612.83) -- (345.57,1628) ;
\draw [shift={(347.5,1628.52)}, rotate = 194.97] [fill={rgb, 255:red, 0; green, 0; blue, 0 }  ][line width=0.08]  [draw opacity=0] (7.2,-1.8) -- (0,0) -- (7.2,1.8) -- cycle    ;
\draw    (287.17,1615.5) -- (347.23,1647.07) ;
\draw [shift={(349,1648)}, rotate = 207.73] [fill={rgb, 255:red, 0; green, 0; blue, 0 }  ][line width=0.08]  [draw opacity=0] (7.2,-1.8) -- (0,0) -- (7.2,1.8) -- cycle    ;
\draw    (289,1639.83) -- (345,1639.83) ;
\draw [shift={(347,1639.83)}, rotate = 180] [fill={rgb, 255:red, 0; green, 0; blue, 0 }  ][line width=0.08]  [draw opacity=0] (7.2,-1.8) -- (0,0) -- (7.2,1.8) -- cycle    ;
\draw    (288.17,1636.83) -- (345.56,1622.02) ;
\draw [shift={(347.5,1621.52)}, rotate = 165.53] [fill={rgb, 255:red, 0; green, 0; blue, 0 }  ][line width=0.08]  [draw opacity=0] (7.2,-1.8) -- (0,0) -- (7.2,1.8) -- cycle    ;
\draw    (287.17,1633.83) -- (347.2,1604.88) ;
\draw [shift={(349,1604.02)}, rotate = 154.26] [fill={rgb, 255:red, 0; green, 0; blue, 0 }  ][line width=0.08]  [draw opacity=0] (7.2,-1.8) -- (0,0) -- (7.2,1.8) -- cycle    ;
\draw    (288.83,1643.16) -- (343.04,1654.12) ;
\draw [shift={(345,1654.52)}, rotate = 191.43] [fill={rgb, 255:red, 0; green, 0; blue, 0 }  ][line width=0.08]  [draw opacity=0] (7.2,-1.8) -- (0,0) -- (7.2,1.8) -- cycle    ;
\draw    (287.17,1645.83) -- (342.63,1667.29) ;
\draw [shift={(344.5,1668.02)}, rotate = 201.15] [fill={rgb, 255:red, 0; green, 0; blue, 0 }  ][line width=0.08]  [draw opacity=0] (7.2,-1.8) -- (0,0) -- (7.2,1.8) -- cycle    ;
\draw    (426,1518.83) -- (490.8,1518.82) ;
\draw [shift={(492.8,1518.82)}, rotate = 179.99] [fill={rgb, 255:red, 0; green, 0; blue, 0 }  ][line width=0.08]  [draw opacity=0] (7.2,-1.8) -- (0,0) -- (7.2,1.8) -- cycle    ;
\draw    (424.83,1545) -- (490.94,1521.99) ;
\draw [shift={(492.83,1521.33)}, rotate = 160.81] [fill={rgb, 255:red, 0; green, 0; blue, 0 }  ][line width=0.08]  [draw opacity=0] (7.2,-1.8) -- (0,0) -- (7.2,1.8) -- cycle    ;
\draw    (425.17,1575) -- (492.23,1524.86) ;
\draw [shift={(493.83,1523.66)}, rotate = 143.22] [fill={rgb, 255:red, 0; green, 0; blue, 0 }  ][line width=0.08]  [draw opacity=0] (7.2,-1.8) -- (0,0) -- (7.2,1.8) -- cycle    ;
\draw    (424.5,1604.33) -- (494.16,1527.48) ;
\draw [shift={(495.5,1526)}, rotate = 132.19] [fill={rgb, 255:red, 0; green, 0; blue, 0 }  ][line width=0.08]  [draw opacity=0] (7.2,-1.8) -- (0,0) -- (7.2,1.8) -- cycle    ;
\draw    (424.5,1634.33) -- (496.7,1529.64) ;
\draw [shift={(497.83,1528)}, rotate = 124.59] [fill={rgb, 255:red, 0; green, 0; blue, 0 }  ][line width=0.08]  [draw opacity=0] (7.2,-1.8) -- (0,0) -- (7.2,1.8) -- cycle    ;
\draw  [fill={rgb, 255:red, 255; green, 255; blue, 255 }  ,fill opacity=1 ][dash pattern={on 0.84pt off 2.51pt}] (315.75,1521.54) .. controls (315.75,1513.81) and (322.02,1507.54) .. (329.75,1507.54) -- (371.75,1507.54) .. controls (379.48,1507.54) and (385.75,1513.81) .. (385.75,1521.54) -- (385.75,1654.02) .. controls (385.75,1661.75) and (379.48,1668.02) .. (371.75,1668.02) -- (329.75,1668.02) .. controls (322.02,1668.02) and (315.75,1661.75) .. (315.75,1654.02) -- cycle ;

\draw (131,1479) node [anchor=north west][inner sep=0.75pt]   [align=left] {Time slot $\displaystyle 1$};
\draw (242,1479.33) node [anchor=north west][inner sep=0.75pt]   [align=left] {Time slot $\displaystyle 2$};
\draw (380.33,1480) node [anchor=north west][inner sep=0.75pt]   [align=left] {Time slot $\displaystyle T$};
\draw (55,1660) node [anchor=north west][inner sep=0.75pt]    {$\vdots $};
\draw (167,1660) node [anchor=north west][inner sep=0.75pt]    {$\vdots $};
\draw (277,1660) node [anchor=north west][inner sep=0.75pt]    {$\vdots $};
\draw (415,1660) node [anchor=north west][inner sep=0.75pt]    {$\vdots $};
\draw (338,1479.65) node [anchor=north west][inner sep=0.75pt]    {$\cdots $};
\draw (343,1530) node [anchor=north west][inner sep=0.75pt]    {$\cdots $};
\draw (343,1580) node [anchor=north west][inner sep=0.75pt]    {$\cdots $};
\draw (343,1630) node [anchor=north west][inner sep=0.75pt]    {$\cdots $};
\draw (496,1513) node [anchor=north west][inner sep=0.75pt]    {$\Omega $};
\draw (53, 1513) node [anchor=north west][inner sep=0.75pt]    {$\Phi$};
\draw (161, 1511) node [anchor=north west][inner sep=0.75pt]    {$\lambda _{1}^{1}$};
\draw (161, 1541) node [anchor=north west][inner sep=0.75pt]    {$\lambda _{2}^{1}$};
\draw (161, 1571) node [anchor=north west][inner sep=0.75pt]    {$\lambda _{3}^{1}$};
\draw (271, 1511) node [anchor=north west][inner sep=0.75pt]    {$\lambda _{1}^{2}$};
\draw (271,1541) node [anchor=north west][inner sep=0.75pt]    {$\lambda _{2}^{2}$};
\draw (271,1571) node [anchor=north west][inner sep=0.75pt]    {$\lambda _{3}^{2}$};
\draw (271,1601) node [anchor=north west][inner sep=0.75pt]    {$\lambda _{4}^{2}$};
\draw (271,1632) node [anchor=north west][inner sep=0.75pt]    {$\lambda _{5}^{2}$};
\draw (407, 1511) node [anchor=north west][inner sep=0.75pt]    {$\lambda _{1}^{T}$};
\draw (407, 1541) node [anchor=north west][inner sep=0.75pt]    {$\lambda _{2}^{T}$};
\draw (407, 1571) node [anchor=north west][inner sep=0.75pt]    {$\lambda _{3}^{T}$};
\draw (407, 1601) node [anchor=north west][inner sep=0.75pt]    {$\lambda _{4}^{T}$};
\draw (407, 1632) node [anchor=north west][inner sep=0.75pt]    {$\lambda _{5}^{T}$};

\end{tikzpicture}

    \end{center}
    \caption{The graph for which finding the shortest path with resource constraint gives the optimum of the AAP.}\label{fig:shortest-path}
\end{figure*}
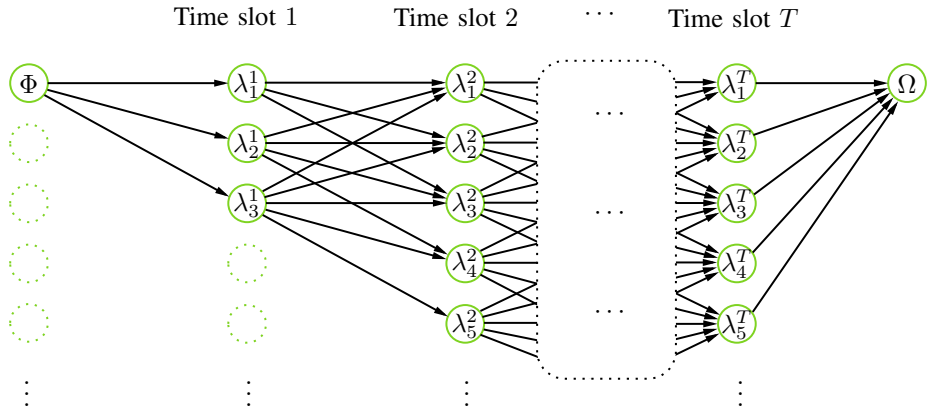

%% file: result1.tex
\begin{figure}[t]
        \begin{center}
        \scriptsize
    		\begin{tikzpicture}
        		\begin{axis}[
        			scaled y ticks=base 10:-3,
        		    xlabel={The number of time slots $T$},
        		    ylabel={Rate per time slot (Gbit/s)},
        		    xmin=6, xmax=14,
        		   xtick={6, 8, 10, 12, 14 },
                        legend style={at={(0.75, 0.5)},anchor=west},
        		    grid style=densely dashed,
        		    tick label style={font=\scriptsize},
        		    label style={font=\small},
        		    legend style={font=\scriptsize},
        		]
        		
            		\addplot[ color=red, mark=square, line width=0.8pt]     
            		coordinates { 
                    ( 6 , 3878)
                    ( 8 , 3678 )
                    ( 10 , 3478 )
                    ( 12 , 3275 )
                    ( 14 , 3120 )
            		};

            		\addplot[ color=black, mark=triangle, densely dashed, mark options={solid}, line width=0.8pt]
            		coordinates {
                    ( 6 , 0 )
                    ( 8 , 0 )
                    ( 10 , 0 )
                    ( 12 , 0 )
                    ( 14 , 0 )
            		};
            		
            		\addplot[ color=blue, mark=o, dotted, mark options={solid}, line width=0.8pt]
            		coordinates {
                    ( 6 , 3528 )
                    ( 8 , 3228 )
                    ( 10 , 3000 )
                    ( 12 , 2808 )
                    ( 14 , 2628 )
            		};

            		\legend{RA, SA, AA}
        		
        		\end{axis}
    		\end{tikzpicture}
    		
        \end{center}
    \caption{Rate per time slot under worst-case performance with respect to the number of time slots $T$.}\label{fig:result1}
\end{figure}
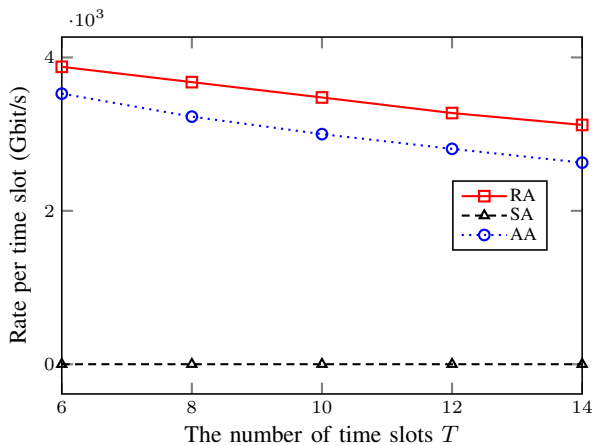

%% file: result2.tex
\begin{figure}
\centering
\scriptsize
\begin{tikzpicture}
  \begin{axis}
    [
    scaled y ticks=base 10:-3,
    clip=false,
    boxplot/draw direction=y,
    boxplot/variable width,
    boxplot/every median/.style={black,very thick,solid},
    ylabel style={align=center}, 
    ylabel={\small Sum rate (Gbit/s)},
    ytick={0,3,6,9,12,15},
    y tick label style={align=right},
    yticklabels={0,3, 6, 9, 12, 15},
    xtick={0,1,2,3,4},
    x tick label style={align=center},
    xticklabels={,RA,SA,AA}
    ]



    \addplot[fill = red, draw = red, fill opacity=0.1,
    mark=*,
    boxplot,
    mark options={fill=white}, 
    boxplot prepared={
      average=8.01,
      lower whisker=6.12,
      lower quartile=6.57,
      median=7.23,
      upper quartile=8.4,
      upper whisker=9.33,
      sample size=3
    },
    ] 
    coordinates {
    }
    node[right,font=\scriptsize, fill opacity=1] at (boxplot box cs: \boxplotvalue{average}, 0.95)
    {\boxplotvalue{average}};    
    ;

    \addplot[fill = black, draw = black, fill opacity=0.1,
    mark=*, 
    boxplot,
    mark options={draw = gray, fill = white}, 
    boxplot prepared={
      average=7.7,
      lower whisker=10.4,
      lower quartile=9.0,
      median=7.2,
      upper quartile=5.8,
      upper whisker=4.45,
      sample size=3 
    },
    ] 
    coordinates {
    (2,0) 
    (2,0.1) 
    (2,0.3) 
    (2,0.9) 
    (2,0.8) 
    (2,3.3) 
    (2,3.0) 
    (2,2.0) 
    (2,15) 
    (2,14.6) 
    (2,13.5) 
    (2,13.2) 
    (2,12.9) 
    (2,10.9) 
    }
    node[right,font=\scriptsize, fill opacity=1] at (boxplot box cs: \boxplotvalue{average}, 0.95)
    {\boxplotvalue{average}};    
    ;

    \addplot[fill = blue, draw = blue, fill opacity=0.1,
    mark=*, 
    boxplot,
    mark options={draw = white!50!blue, fill=white}, 
    boxplot prepared={
      average=8.34,
      lower whisker=6.3,
      lower quartile=6.84,
      median=7.2,
      upper quartile=8.7,
      upper whisker=9.45,
      sample size=3 
    },
    ] 
    coordinates {
    (3, 9.96)
    (3, 10.2)
    (3, 5.1)
    (3, 5.3)
    (3, 5.6)
    (3, 4.8)
    }
    node[right,font=\scriptsize, fill opacity=1] at (boxplot box cs: \boxplotvalue{average}, 0.95)
    {\boxplotvalue{average}};    
    ;

  \end{axis}
\end{tikzpicture}
\caption{Box plot showing the sum rate results obtained from 1000 randomly generated cases with $T = 10$. The numbers on the right of the boxes indicate the average sum rate for each individual scheme.}\label{fig:result2}    
\end{figure}
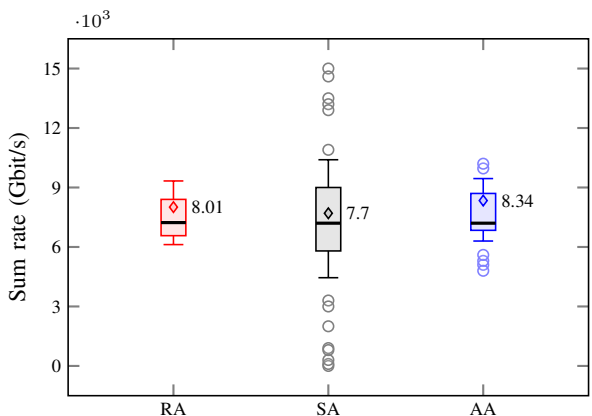

%% file: result3.tex
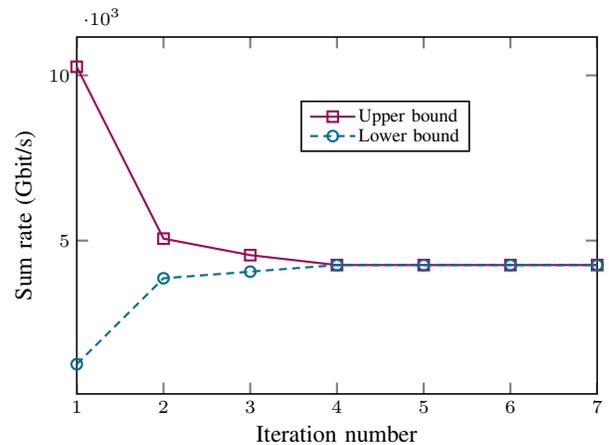
\begin{figure}[t]
        \begin{center}
        \scriptsize
    		\begin{tikzpicture}
        		\begin{axis}[
        			scaled y ticks=base 10:-3,
        		    xlabel={Iteration number},
        		    ylabel={Sum rate (Gbit/s)},
        		    xmin=1, xmax=7,
        		   xtick={1, 2, 3, 4, 5, 6, 7},
        		    legend style={at={(0.43, 0.75)},anchor=west},
        		    grid style=densely dashed,
        		    tick label style={font=\scriptsize},
        		    label style={font=\small},
        		    legend style={font=\scriptsize},
        		]
            		\addplot[ color={rgb:red,255;green,20;blue,147}, mark = square, line width=0.8pt]     
            		coordinates { 
                    ( 1 , 10260 )
                    ( 2 , 5060 )
                    ( 3 , 4560 )
                    ( 4 , 4260 )
                    ( 5 , 4260 )
                    ( 6 , 4260 )
                    ( 7 , 4260 )
            		};
            		\addplot[ color={rgb:red,0;green,178;blue,238}, mark = o, mark options={solid}, densely dashed, line width=0.8pt]
            		coordinates {
                    ( 1 , 1260 )
                    ( 2 , 3860 )
                    ( 3 , 4060 )
                    ( 4 , 4260 )
                    ( 5 , 4260 )
                    ( 6 , 4260 )
                    ( 7 , 4260 )
            		};

            		\legend{Upper bound, Lower bound}
        		
        		\end{axis}
    		\end{tikzpicture}
    		
        \end{center}
    \caption{Convergence behavior of the proposed approach with upper and lower bounds as functions of the number of iterations when $T = 8$.}\label{fig:result3}
\end{figure}

%% file: main.bbl
\begin{thebibliography}{10}
\providecommand{\url}[1]{#1}
\csname url@samestyle\endcsname
\providecommand{\newblock}{\relax}
\providecommand{\bibinfo}[2]{#2}
\providecommand{\BIBentrySTDinterwordspacing}{\spaceskip=0pt\relax}
\providecommand{\BIBentryALTinterwordstretchfactor}{4}
\providecommand{\BIBentryALTinterwordspacing}{\spaceskip=\fontdimen2\font plus
\BIBentryALTinterwordstretchfactor\fontdimen3\font minus
  \fontdimen4\font\relax}
\providecommand{\BIBforeignlanguage}[2]{{%
\expandafter\ifx\csname l@#1\endcsname\relax
\typeout{** WARNING: IEEEtran.bst: No hyphenation pattern has been}%
\typeout{** loaded for the language `#1'. Using the pattern for}%
\typeout{** the default language instead.}%
\else
\language=\csname l@#1\endcsname
\fi
#2}}
\providecommand{\BIBdecl}{\relax}
\BIBdecl

\bibitem{Toy2021}
M.~Toyoshima, ``Recent trends in space laser communications for small
  satellites and constellations,'' \emph{Journal of Lightwave Technology},
  vol.~39, no.~3, pp. 693--699, 2021.

\bibitem{Gue2004}
M.~Guelman, A.~Kogan, A.~Kazarian, A.~Livne, M.~Orenstein, and H.~Michalik,
  ``Acquisition and pointing control for inter-satellite laser
  communications,'' \emph{IEEE Transactions on Aerospace and Electronic
  Systems}, vol.~40, no.~4, pp. 1239--1248, 2004.

\bibitem{song2016performance}
T.~Song, Q.~Wang, M.-W. Wu, and P.-Y. Kam, ``Performance of laser
  inter-satellite links with dynamic beam waist adjustment,'' \emph{Optics
  Express}, vol.~24, no.~11, pp. 11\,950--11\,960, 2016.

\bibitem{ben2009robust}
A.~Ben-Tal, L.~El~Ghaoui, and A.~Nemirovski, \emph{Robust Optimization}.\hskip
  1em plus 0.5em minus 0.4em\relax Princeton university press, 2009.

\bibitem{lee2022dynamic}
K.~Lee, V.~Mai, and H.~Kim, ``Dynamic adaptive beam control system using
  variable focus lenses for laser inter-satellite link,'' \emph{IEEE Photonics
  Journal}, vol.~14, no.~4, pp. 1--8, 2022.

\bibitem{mai2019beam}
V.~V. Mai and H.~Kim, ``Beam size optimization and adaptation for high-altitude
  airborne free-space optical communication systems,'' \emph{IEEE Photonics
  Journal}, vol.~11, no.~2, pp. 1--13, 2019.

\bibitem{toyoshima2002optimum}
M.~Toyoshima, T.~Jono, K.~Nakagawa, and A.~Yamamoto, ``Optimum divergence angle
  of a {Gaussian} beam wave in the presence of random jitter in free-space
  laser communication systems.'' \emph{Journal of the Optical Society of
  America A}, vol.~19, no.~3, pp. 567--571, 2002.

\bibitem{do2020numerical}
P.~X. Do, A.~Carrasco-Casado, T.~Van~Vu, T.~Hosonuma, M.~Toyoshima, and
  S.~Nakasuka, ``Numerical and analytical approaches to dynamic beam waist
  optimization for leo-to-geo laser communication,'' \emph{OSA Continuum},
  vol.~3, no.~12, pp. 3508--3522, 2020.

\bibitem{saleh2019fundamentals}
B.~E. Saleh and M.~C. Teich, \emph{Fundamentals of Photonics}.\hskip 1em plus
  0.5em minus 0.4em\relax John Wiley \& Sons, Inc., 2019.

\bibitem{gorissen2015practical}
B.~L. Gorissen, {\.I}.~Yan{\i}ko{\u{g}}lu, and D.~den Hertog, ``A practical
  guide to robust optimization,'' \emph{{Omega}}, vol.~53, pp. 124--137, 2015.

\bibitem{majumdar2005free}
A.~K. Majumdar, ``Free-space laser communication performance in the atmospheric
  channel,'' \emph{Journal of Optical and Fiber Communications Reports},
  vol.~2, no.~4, pp. 345--396, 2005.

\bibitem{averbakh2008explicit}
I.~Averbakh and Y.-B. Zhao, ``Explicit reformulations for robust optimization
  problems with general uncertainty sets,'' \emph{SIAM Journal on
  Optimization}, vol.~18, no.~4, pp. 1436--1466, 2008.

\bibitem{JOKSCH1966191}
``The shortest route problem with constraints,'' \emph{Journal of Mathematical
  Analysis and Applications}, vol.~14, no.~2, pp. 191--197, 1966.

\end{thebibliography}
